\documentclass[11pt]{article}

\usepackage{amsmath,amsfonts,amsthm}
\usepackage{booktabs}
\usepackage{hyperref}
\usepackage{fullpage}
\usepackage[numbers,sort&compress]{natbib}

\theoremstyle{plain}
\newtheorem{theorem}{Theorem}[section]
\newtheorem{lemma}[theorem]{Lemma}

\theoremstyle{definition}
\newtheorem{definition}[theorem]{Definition}

\newtheorem{remark}[theorem]{Remark}

\newcommand{\R}{\mathbb{R}}
\newcommand{\LW}{\mathrm{LW}}
\newcommand{\OPT}{\mathrm{OPT}}
\newcommand{\ALG}{\mathrm{ALG}}
\newcommand{\PoA}{\mathrm{PoA}}
\newcommand{\SPA}{\mathrm{SPA}}
\newcommand{\FPA}{\mathrm{FPA}}
\newcommand{\rFPA}{\mathrm{rFPA}}
\newcommand{\rTruth}{\mathrm{rTruth}}
\newcommand{\pFPA}{\mathrm{pFPA}}

\begin{document}

\title{Efficiency of Generalized Proportional First-Price Auctions Under Auto-bidding\thanks{The main result was entirely obtained by Cogentic, an agentic framework for mathematical discovery, using an internal version of Gemini as the base model. The authors contextualized the findings and verified the proofs. The full exposition here is due to the authors aided by different AI models.

The following authors have additional affiliations beyond Google Research: Yang Cai (Yale University) and Vineet Gupta (Google DeepMind).}}

\author{
  Yang Cai, Vineet Gupta, Yanchen Jiang, Christopher Liaw \\[0.5ex]
  Aranyak Mehta, Grigoris Velegkas, and Di Wang \\[1.5ex]
  Google Research \\
  \small\texttt{\{caiy, vineet, yanchenjiang, cvliaw, aranyak, gvelegkas, wadi\}@google.com}
}
\date{}

\maketitle

\begin{abstract}
Auto-bidding is now widely adopted in online advertising platforms, allowing advertisers to specify high-level campaign objectives—such as maximizing total value subject to a return-on-spend (ROS) constraint—rather than manual per-query bids. A central question in algorithmic mechanism design is characterizing the worst-case efficiency loss, or Price of Anarchy (PoA), across auction formats in this prior-free setting. While randomized auctions are known to strictly improve efficiency over deterministic mechanisms for two bidders, two fundamental questions have remained open: (1) what is the optimal PoA for two bidders, and (2) can any mechanism beat the barrier of 2 for general $n \ge 3$ bidders?

We resolve both questions using the family of $r$-proportional first-price auctions ($\text{pFPA}_r$), in which each bidder wins with probability proportional to their bid raised to an exponent $r > 0$ and pays their bid upon winning. First, for two bidders, we prove that the standard proportional first-price auction ($r = 1$) achieves a tight $\text{PoA} \le 1.5$, complemented by a matching lower bound showing that no anonymous, monotone mechanism can do better. Second, for general $n \ge 2$ bidders, setting $r = 2n$ achieves $\text{PoA} \le 2 - \frac{1}{4n+1} = 2 - \Omega(1/n)$ across all undominated bid profiles, breaking the deterministic barrier of 2 for every finite $n$ and asymptotically matching the known $2 - \Theta(1/n)$ lower bound.
\end{abstract}

\section{Introduction}
\label{sec:intro}

Automated bidding (auto-bidding) has become the predominant interface for advertisers participating in online advertising markets, including sponsored search and display advertising \cite{google2022autobidding, meta2022bid}. Rather than manually calibrating fine-grained bids for individual ad opportunities (queries), advertisers specify high-level campaign objectives and portfolio-wide constraints. An automated bidding agent then translates these aggregate goals into real-time, per-query bids. A foundational and widely deployed auto-bidding paradigm is \emph{value maximization subject to a return-on-spend (ROS) constraint} (also referred to as a return-on-investment or ROI constraint), in which an advertiser seeks to maximize the total value of won impressions subject to keeping their average spend per unit of acquired value below a prescribed target threshold \cite{aggarwal2019autobidding, deng2021towards, balseiro2021landscape}.

Because auto-bidders couple decisions across all queries through their global ROS constraints and aim to maximize total value rather than classical quasilinear utility, the equilibrium behavior and welfare guarantees of standard auctions change fundamentally in the auto-bidding world. In their seminal work, Aggarwal et al.~\cite{aggarwal2019autobidding} formulated the mathematical model of auto-bidding and evaluated system efficiency via the \emph{liquid welfare} \cite{dobzinski2014efficiency}---the aggregate target-weighted value of an allocation. They proved that when the underlying auction is truthful (such as the second-price auction, $\SPA$), uniform bidding (bidding a fixed multiplier of one's value across all queries) is a dominant strategy, and the resulting Price of Anarchy ($\PoA$) of $\SPA$ is exactly $2$.

Subsequent research investigated whether leaving the class of deterministic truthful mechanisms can beat the $\PoA$ barrier of $2$ in the prior-free model (where the mechanism observes only bids and has no access to private values). Liaw, Mehta, and Perlroth~\cite{liaw2023efficiency} also a proved dominance theorem for deterministic auctions: any uniform-bidding equilibrium of $\SPA$ with tight ROS constraints can be replicated in any reasonable deterministic auction (such as the first-price auction, $\FPA$), and consequently \emph{no} deterministic auction can achieve a $\PoA$ better than $2$ (while $\FPA$ with general non-uniform bidding matches the $\PoA$ of $2$).

When randomization is permitted, however, the landscape changes dramatically. For the two-bidder setting, Mehta~\cite{mehta2022auction} designed a randomized truthful auction ($\rTruth$) that achieves a $\PoA$ of approximately $1.89$. Liaw, Mehta, and Perlroth~\cite{liaw2023efficiency} subsequently demonstrated that \emph{non-truthful} (first-price) payments amplify the benefits of randomization: by introducing a threshold-based randomized first-price auction ($\rFPA(\alpha)$), they improved the two-bidder $\PoA$ to $1.8$ (at $\alpha = 1.4$). At the same time, Liaw, Mehta, and Perlroth~\cite[Theorem~6.2]{liaw2023efficiency} proved a lower bound showing that under mild assumptions, for any $n = 2k$ bidders, any anonymous mechanism has $\PoA \ge \bigl(\frac{k+2}{2(k+1)}\bigr)^{-1} = \frac{2(k+1)}{k+2} = \frac{2n+4}{n+4} = 2 - \frac{4}{n+4}$. Yet two fundamental questions remained unresolved:
\begin{quote}
\emph{(1) For $n = 2$ bidders, what is the optimal Price of Anarchy across prior-free mechanisms? \\
(2) For general $n \ge 3$ bidders, does there exist a mechanism that breaks the deterministic $\PoA$ barrier of $2$ and matches the lower bound rate of $2 - \Theta(1/n)$?}
\end{quote}

\subsection{Our Results}
\label{subsec:results}

In this paper, we resolve both open questions using the family of \emph{$r$-Proportional First-Price Auctions} ($\pFPA_r$). Given non-negative bids $\mathbf{b}_j = (b_{1,j}, \dots, b_{n,j}) \in \R_{\ge 0}^n$ on query $j$, each bidder $i \in [n]$ wins query $j$ with probability proportional to $b_{i,j}^r$,
\[
  x_{i,j}(\mathbf{b}_j) = \frac{b_{i,j}^r}{\sum_{k=1}^n b_{k,j}^r},
\]
and pays their bid $b_{i,j}$ conditional on winning (and $0$ if uncontested), resulting in an expected payment of $p_{i,j}(\mathbf{b}_j) = b_{i,j} x_{i,j}(\mathbf{b}_j)$ on contested queries.

\begin{table}[t]
\centering
\caption{Summary of previous and new Price of Anarchy ($\PoA$) bounds in the prior-free auto-bidding setting.}
\label{tab:summary}
\begin{tabular}{lll}
\toprule
\textbf{Auction Setting} & \textbf{PoA} & \textbf{Reference} \\
\midrule
Second-Price Auction, $\SPA$ ($n \ge 2$) & $= 2$ & \cite{aggarwal2019autobidding} \\
First-Price Auction, $\FPA$ ($n \ge 2$) & $= 2$ & \cite{liaw2023efficiency, deng2022efficiency} \\
Randomized Truthful, $\rTruth$ ($n = 2$) & $\le 1.89$ & \cite{mehta2022auction} \\
Randomized Threshold FPA, $\rFPA(\alpha=1.4)$ ($n = 2$) & $\le 1.8$ & \cite{liaw2023efficiency} \\
\textbf{Proportional $\pFPA_1$ ($n = 2$)} & $\mathbf{= 1.5}$ & \textbf{Theorem~\ref{thm:main}} \\
\textbf{Two-Bidder Lower Bound ($n = 2$)} & $\mathbf{\ge 1.5}$ & \textbf{Theorem~\ref{thm:lower_bound}} \\
\textbf{$2n$-Proportional $\pFPA_{2n}$ ($n \ge 2$)} & $\mathbf{\le 2 - \frac{1}{4n+1}}$ & \textbf{Theorem~\ref{thm:n_bidders}} \\
Multi-Bidder Lower Bound ($n = 2k$) & $\ge 2 - \frac{4}{n+4}$ & \cite[Theorem~6.2]{liaw2023efficiency} \\
\bottomrule
\end{tabular}
\end{table}

As summarized in Table~\ref{tab:summary}, we establish three main results.

\paragraph{Two-Bidder Upper Bound (Section~\ref{sec:upper_bound}).}
Our first main result establishes that the standard two-bidder Proportional First-Price Auction ($\pFPA_1$, with $n = 2$ and $r = 1$) achieves a Price of Anarchy of at most $1.5$.

\begin{theorem}[Formal version in Theorem~\ref{thm:main}]
\label{thm:intro_upper}
In the two-bidder Proportional First-Price Auction ($\pFPA_1$), the Price of Anarchy is at most $1.5$.
\end{theorem}

To prove Theorem~\ref{thm:intro_upper}, we start by writing each bidder's value-maximization problem subject to their return-on-spend constraint and analyzing its Lagrangian dual (Lemma~\ref{lem:binding_ros}). This decouples the problem across queries into an unconstrained surrogate objective of the form $\sum_{j \in Q} \bigl(x_{i,j}(b_j) v_{i,j} - \mu_i p_{i,j}(b_j)\bigr)$, where $\mu_i = \frac{\lambda_i}{1 + \lambda_i} \in [0, 1)$ is the dual variable corresponding to bidder $i$'s global return-on-spend constraint. Note that both the allocation probability $x_{i,j}(b_j)$ and expected payment $p_{i,j}(b_j)$ depend on the submitted bids. The first-order optimality condition of this surrogate objective yields an explicit constraint relating values, bids, and the dual multipliers: $v_{i,j} = \mu_i b_{i,j} \bigl(\frac{b_{i,j}}{b_{-i,j}} + 2\bigr)$ (Lemma~\ref{lem:foc}). Writing $\Delta_{i,j} = x_{i,j} v_{i,j} - p_{i,j}$ for bidder $i$'s expected net surplus (query slack) on query $j$, we use this relation to derive a pointwise inequality for every query $j \in Q$ and any $\mu_1, \mu_2 \ge 0$ (Lemma~\ref{lem:pointwise}):
\[
\max(v_{1,j}, v_{2,j}) \le 1.5 (p_{1,j} + p_{2,j}) + (1 + 3\mu_1)\Delta_{1,j} + (1 + 3\mu_2)\Delta_{2,j}.
\]
Finally, complementary slackness guarantees that $\mu_i \sum_{j \in Q} \Delta_{i,j} = 0$ (Lemma~\ref{lem:binding_ros}), where $\mu_i$ is the dual variable corresponding to the overall constraint and $\sum_{j \in Q} \Delta_{i,j} = V_i - S_i$, where $V_i, S_i$ are the value and spend respectively, represents bidder $i$'s aggregate slack. Summing the pointwise inequality over all queries $j \in Q$ and applying $\mu_i \sum_{j \in Q} \Delta_{i,j} = 0$ together with the ROS feasibility condition $S_i \le V_i$ yields $\LW(\OPT) \le \sum_{i=1}^2 \bigl(V_i + \frac{1}{2}S_i\bigr) \le 1.5\,\LW(\ALG)$, establishing the claim.

\paragraph{Two-Bidder Lower Bound (Section~\ref{sec:lower_bound}).}
Our second main result shows that our $1.5$ upper bound for $\pFPA_1$ cannot be improved by any two-bidder mechanism satisfying mild assumptions.

\begin{theorem}[Formal version in Theorem~\ref{thm:lower_bound}]
\label{thm:intro_lower}
For any anonymous two-bidder mechanism $\mathcal{M}$ satisfying monotone allocation and monotone conditional price at a symmetric tie, $\PoA(\mathcal{M}) \ge 1.5$.
\end{theorem}

We prove Theorem 1.2 by adapting the lower bound construction of Liaw, Mehta, and Perlroth~\cite[Theorem~6.2]{liaw2023efficiency} (which gave a lower bound of $4/3$ for two bidders). The high-level intuition is that bidder 1 has value $1/2$ for an uncontested query $A$ and essentially zero value for query $B$, while bidder 2 has value $1$ for query $B$ and zero value for query $A$. Bidder 1 wins query $A$ for free, accumulating $1/2$ in return-on-spend surplus (slack). Bidder 1 then spends this surplus to tie bidder 2 on query $B$, splitting it equally: each bidder wins query $B$ with probability $1/2$ and pays $1/2$. In this equilibrium, bidder 1 receives $1/2$ value from query $A$ and bidder 2 receives $1/2$ value from query $B$, yielding $\LW(\ALG) = 1$. In contrast, $\OPT$ allocates query $A$ to bidder 1 and query $B$ to bidder 2 for total welfare $\LW(\OPT) = 1.5$, giving a tight $\PoA = 1.5$.

Note that to make this argument go through, we make use a third query that provides ``slack'' for bidder $2$ to win query $B$ and ensure that query $B$ also uses up slack for bidder $2$.
This is more for technical reasons to ensure bidder $2$ cannot win more of query $B$ while maintaining their constraints.

\paragraph{Tight $2 - \Theta(1/n)$ Scaling for $n \ge 2$ Bidders (Section~\ref{sec:n_bidders}).}
Our third main result shows that the $2n$-Proportional First-Price Auction ($\pFPA_{2n}$, with $r = 2n$) breaks the deterministic barrier of $2$ for every $n \ge 2$ and matches the optimal $2 - \Theta(1/n)$ scaling.

\begin{theorem}[Formal version in Theorem~\ref{thm:n_bidders}]
\label{thm:intro_n_bidders}
For any $n \ge 2$ bidders, in the $2n$-Proportional First-Price Auction ($\pFPA_{2n}$), any undominated bid profile satisfies $\frac{\LW(\OPT)}{\LW(\ALG)} \le 2 - \frac{1}{4n + 1}$, and hence $\PoA(\pFPA_{2n}) \le 2 - \frac{1}{4n + 1}$.
\end{theorem}

For general $n \ge 2$ bidders, we analyze $\pFPA_r$ under any set of undominated bids by adapting the undominated-bid and value-plus-spend convex-combination framework of Liaw, Mehta, and Perlroth~\cite[Section~5]{liaw2023efficiency} and Mehta~\cite{mehta2022auction} to $n$ bidders. We first show that any undominated bid satisfies $b_{i,j} \ge \frac{r(1 - x_{i,j})}{r(1 - x_{i,j}) + 1} v_{i,j}$ (Lemma~\ref{lem:n_undom}), and use Jensen's inequality on the $n - 1$ non-optimal bidders' winning probabilities to lower-bound the total query spend $S^{(j)}$ as a function of the optimal bidder's winning probability $x^*_j$ (Lemma~\ref{lem:n_spend}). Reducing the resulting minimax program to the characteristic equation $(n - 1)c^{r+1} + c - r = 0$ (Lemmas~\ref{lem:n_duality}--\ref{lem:n_root}), we prove that setting $r = 2n$ yields $c > 1 + \frac{1}{4n}$ and hence $\PoA(\pFPA_{2n}) \le 1 + \frac{1}{c} < 2 - \frac{1}{4n + 1} = 2 - \Omega(1/n)$ (Theorem~\ref{thm:n_bidders}), matching the $2 - \frac{4}{n+4} = 2 - O(1/n)$ lower bound of \cite{liaw2023efficiency} for all $n \ge 2$, up to the constant on the $1/n$ term.

\subsection{Related Work}
\label{subsec:related}
Our work contributes to the growing literature on auction design, equilibrium efficiency, and learning dynamics under auto-bidding; see Aggarwal et al.~\cite{aggarwal2024autobidding} for a comprehensive survey.
\paragraph{Prior-Free Auction Efficiency under Auto-Bidding.}
The theoretical study of auto-bidding was initiated by Aggarwal et al.~\cite{aggarwal2019autobidding}, who showed that under return-on-spend (ROS) and budget constraints, uniform bidding is optimal in truthful auctions and yields a tight liquid Price of Anarchy ($\PoA$) of $2$ for $\SPA$ (see also Babaioff et al.~\cite{babaioff2021non}). In the prior-free setting without side information, Mehta~\cite{mehta2022auction} showed that randomization improves the two-bidder $\PoA$ to $\approx 1.89$ via a randomized truthful auction ($\rTruth$), while proving that any truthful auction has $\PoA \ge 2$ as $n \to \infty$. Liaw, Mehta, and Perlroth~\cite{liaw2023efficiency} extended this study to non-truthful mechanisms: they proved that no deterministic auction beats $\PoA = 2$, established $\PoA(\FPA) = 2$ under non-uniform bidding (also proved concurrently by Deng et al.~\cite{deng2022efficiency} and further refined for slice-based and hybrid bidding in \cite{deng2024autobidders, colinibaldeschi2026optimal}), introduced a randomized first-price auction ($\rFPA$) achieving $\PoA \le 1.8$ for two bidders, and proved a lower bound of $\PoA \ge \frac{2n+4}{n+4} = 2 - \frac{4}{n+4}$ for $n = 2k$ bidders. Prior-free efficiency has also been studied under Generalized Second-Price auctions~\cite{deng2024gsp}, user costs~\cite{deng2023usercosts}, joint budget and ROS constraints~\cite{liaw2024efficiency, colinibaldeschi2026optimal}, and interdependent values~\cite{banchio2025autobidding}. Our work resolves both open questions in the canonical prior-free ROS setting of \cite{mehta2022auction, liaw2023efficiency, aggarwal2024autobidding}: we establish $\PoA(\pFPA_1) = 1.5$ for $n = 2$ (with a matching $1.5$ lower bound under mild assumptions) and $\PoA(\pFPA_{2n}) \le 2 - \frac{1}{4n+1} = 2 - \Omega(1/n)$ for all $n \ge 2$, matching the $2 - O(1/n)$ lower bound of \cite{liaw2023efficiency}.
\paragraph{Side Information, Bayesian Design, and Strategic Extensions.}
When the platform has machine-learned value predictions, augmenting $\SPA$ or $\FPA$ with bid boosts or personalized reserves improves the Price of Anarchy and individual welfare guarantees~\cite{deng2021towards, balseiro2021robust, deng2024individual}, while posted-price and prior-independent mechanisms achieve constant-factor approximations~\cite{deng2022posted}. A parallel Bayesian literature characterizes revenue- and welfare-optimal mechanisms for value-maximizing and ROI-constrained buyers across public and private constraint settings~\cite{golrezaei2021auction, balseiro2021landscape, balseiro2022optimal, balseiro2024optimal}. Beyond single-platform static auctions with fixed constraints, recent work has examined auctions without seller commitment~\cite{perlroth2023auctions}, advertiser-level incentive compatibility in reporting ROS and budget targets~\cite{alimohammadi2023incentive, feng2024strategic}, and multi-channel or multi-platform auto-bidding~\cite{aggarwal2023multi, aggarwal2025multiplatform}.
\paragraph{Pacing, Online Learning, and Proportional Allocation.}
Another active research thread studies the structure and PPAD-completeness of multiplicative pacing equilibria in second-price and first-price markets~\cite{conitzer2022multiplicative, conitzer2022pacing, chen2021complexity}, online bidding algorithms under budget and ROS constraints~\cite{balseiro2019learning, feng2023online, balseiro2024field, vijayan2025online}, and liquid welfare guarantees for pacing and no-regret learning dynamics without equilibrium convergence~\cite{gaitonde2023budget, fikioris2023liquid, lucier2024autobidders, aggarwal2025noregret}. Finally, our efficiency benchmark is the \emph{liquid welfare}~\cite{dobzinski2014efficiency, syrgkanis2013composable}, and our allocation rule builds on proportional sharing (the Kelly mechanism), whose efficiency has been extensively studied for divisible resource allocation with quasilinear and budget-constrained agents~\cite{kelly1997charging, johari2004efficiency, christodoulou2016proportional, caragiannis2016welfare}. 
Unlike divisible resource sharing where every bidder pays their bid unconditionally, $\pFPA_r$ is a single-slot randomized auction in which only the winner pays their bid and cross-query decisions are coupled via global ROS constraints.

\section{Preliminaries}
\label{sec:prelim}

\subsection{Auto-bidding and \texorpdfstring{$r$}{r}-Proportional FPA Model}

Let $A = \{1, \dots, n\}$ denote a set of $n \ge 2$ auto-bidding advertisers and let $Q = \{1, \dots, m\}$ denote a finite, non-empty set of $m \ge 1$ single-slot queries (ad opportunities). For each advertiser $i \in A$ and each query $j \in Q$, let $v_{i,j} \ge 0$ denote the non-negative intrinsic value that advertiser $i$ derives from winning query $j$ (with $\max_{i \in A} v_{i,j} > 0$ for at least one query $j \in Q$), and let $T_i > 0$ denote advertiser $i$'s target return-on-spend (ROS) parameter. For a given advertiser $i \in A$, we write $\mathbf{b}_{-i,j} = (b_{k,j})_{k \ne i}$ for the competing bids on query $j$ (and when $n = 2$, we write $-i = 3 - i$ for the single competing advertiser).

Each query $j \in Q$ is sold independently using the \emph{$r$-Proportional First-Price Auction} ($\pFPA_r$, for an exponent parameter $r > 0$), denoted by $\mathcal{M}_{\mathrm{prop}}^{(r)}$ (and simply $\mathcal{M}_{\mathrm{prop}} = \mathcal{M}_{\mathrm{prop}}^{(1)}$ or $\pFPA$ when $r = 1$).

\begin{definition}[$r$-Proportional First-Price Auction ($\pFPA_r$)]
\label{def:pfpa}
Fix an exponent $r > 0$. For each query $j \in Q$, advertisers $1, \dots, n$ simultaneously submit non-negative bids $\mathbf{b}_j = (b_{1,j}, \dots, b_{n,j}) \in \R_{\ge 0}^n$. Under $\mathcal{M}_{\mathrm{prop}}^{(r)}$:
\begin{enumerate}
  \item \textbf{Allocation rule:} If $\sum_{k=1}^n b_{k,j}^r > 0$, advertiser $i \in A$ wins query $j$ with probability
  \begin{equation}
  \label{eq:alloc}
    x_{i,j}(\mathbf{b}_j) = \frac{b_{i,j}^r}{\sum_{k=1}^n b_{k,j}^r},
  \end{equation}
  whereas if $\mathbf{b}_j = \mathbf{0}$, the query is unallocated ($x_{i,j}(\mathbf{0}) = 0$).
  \item \textbf{Payment rule:} If at least two advertisers submit strictly positive bids ($b_{i,j} > 0$ and $\sum_{k \ne i} b_{k,j} > 0$), the winning advertiser $i$ pays their submitted bid $b_{i,j}$. If only one advertiser submits a positive bid ($b_{i,j} > 0$ and $\sum_{k \ne i} b_{k,j} = 0$), advertiser $i$ wins query $j$ with probability $1$ and pays $0$ (receiving the uncontested query for free). Consequently, advertiser $i$'s expected payment on query $j$ is
  \begin{equation}
  \label{eq:payment}
    p_{i,j}(\mathbf{b}_j) =
    \begin{cases}
      \dfrac{b_{i,j}^{r+1}}{\sum_{k=1}^n b_{k,j}^r} & \text{if } b_{i,j} > 0 \text{ and } \sum_{k \ne i} b_{k,j} > 0, \\[6pt]
      0 & \text{if } b_{i,j} = 0 \text{ or } \sum_{k \ne i} b_{k,j} = 0,
    \end{cases}
  \end{equation}
  so that $p_{i,j}(\mathbf{b}_j) = b_{i,j} x_{i,j}(\mathbf{b}_j)$ on every contested query.
\end{enumerate}
When the bid profile $\mathbf{b} = (b_{i,j})_{i \in A, j \in Q}$ is clear from context, we write $x_{i,j}$ and $p_{i,j}$ for $x_{i,j}(\mathbf{b}_j)$ and $p_{i,j}(\mathbf{b}_j)$, respectively.
\end{definition}

\begin{remark}
\label{rem:uncontested_free}
Having a sole positive bidder ($b_{i,j} > 0, \sum_{k \ne i} b_{k,j} = 0$) pay $0$ in Eq.~\eqref{eq:payment} corresponds to the infimum $\inf_{b > 0} \frac{b^{r+1}}{b^r + 0} = 0$ of the uncontested first-price payment as $b \to 0^+$. Without this rule (if $p_{i,j}(b_{i,j}, 0) = b_{i,j}$), whenever $\sum_{k \ne i} b_{k,j} = 0$ on a query with $v_{i,j} > 0$ and advertiser $i$ has a binding ROS constraint ($\mu_i > 0$), advertiser $i$ would want to lower $b_{i,j} \to 0^+$ as much as possible to keep $x_{i,j} = 1$ while reducing spend, precluding the existence of an equilibrium.
\end{remark}

Fix an advertiser $i \in A$ and suppose the competing advertisers submit $\mathbf{b}_{-i} = (\mathbf{b}_{-i,j})_{j \in Q}$. The auto-bidding agent for advertiser $i$ chooses a bid vector $\mathbf{b}_i = (b_{i,j})_{j \in Q} \in \R_{\ge 0}^{|Q|}$ to maximize advertiser $i$'s total expected value
\[
  V_i(\mathbf{b}_i, \mathbf{b}_{-i}) = \sum_{j \in Q} x_{i,j}(\mathbf{b}_j) v_{i,j}
\]
subject to the global return-on-spend (ROS) constraint that total expected spend $S_i(\mathbf{b}_i, \mathbf{b}_{-i}) = \sum_{j \in Q} p_{i,j}(\mathbf{b}_j)$ does not exceed $T_i V_i(\mathbf{b}_i, \mathbf{b}_{-i})$:
\begin{equation}
\label{eq:autobidder_opt}
\begin{aligned}
  \max_{\mathbf{b}_i \in \R_{\ge 0}^{|Q|}} \quad & V_i(\mathbf{b}_i, \mathbf{b}_{-i}) = \sum_{j \in Q} x_{i,j}(\mathbf{b}_j) v_{i,j} \\
  \text{s.t.} \quad & S_i(\mathbf{b}_i, \mathbf{b}_{-i}) \le T_i V_i(\mathbf{b}_i, \mathbf{b}_{-i}).
\end{aligned}
\end{equation}

\subsection{Liquid Welfare and Price of Anarchy}

To measure the aggregate efficiency of an allocation across auto-bidders with ROS constraints, we adopt the standard notion of \emph{liquid welfare} introduced by Dobzinski and Paes Leme~\cite{dobzinski2014efficiency} and used throughout the auto-bidding literature \cite{aggarwal2019autobidding, mehta2022auction, liaw2023efficiency}.

\begin{definition}[Liquid Welfare]
\label{def:liquid_welfare}
For any randomized allocation $\mathbf{x} = (x_{i,j})_{i \in A, j \in Q}$ satisfying $x_{i,j} \ge 0$ and $\sum_{i \in A} x_{i,j} \le 1$ for all $j \in Q$, the \emph{liquid welfare} is defined as
\begin{equation}
\label{eq:lw_def}
  \LW(\mathbf{x}) = \sum_{i \in A} T_i \sum_{j \in Q} x_{i,j} v_{i,j}.
\end{equation}
\end{definition}

\begin{remark}
\label{rem:normalization}
Following Remark~2.4 of \cite{liaw2023efficiency}, we may assume without loss of generality that $T_i = 1$ for all $i \in A$. Indeed, defining the rescaled values $v'_{i,j} = T_i v_{i,j} \ge 0$, the constraint in Eq.~\eqref{eq:autobidder_opt} becomes $\sum_{j \in Q} p_{i,j} \le \sum_{j \in Q} x_{i,j} v'_{i,j}$, the objective $\sum_{j \in Q} x_{i,j} v_{i,j} = \frac{1}{T_i}\sum_{j \in Q} x_{i,j} v'_{i,j}$ has identical best responses (since $T_i > 0$ is a positive constant), and the liquid welfare in Eq.~\eqref{eq:lw_def} becomes $\LW(\mathbf{x}) = \sum_{i \in A} \sum_{j \in Q} x_{i,j} v'_{i,j}$. Henceforth, we set $T_i = 1$ for all $i \in A$.
\end{remark}

Under the normalization $T_i = 1$, for each advertiser $i \in A$ and query $j \in Q$, we define the \emph{expected net surplus} on query $j$ as
\begin{equation}
\label{eq:surplus_def}
  \Delta_{i,j}(\mathbf{b}_j) = x_{i,j}(\mathbf{b}_j) v_{i,j} - p_{i,j}(\mathbf{b}_j).
\end{equation}
Advertiser $i$'s ROS constraint $S_i \le V_i$ in Eq.~\eqref{eq:autobidder_opt} is then equivalently written as
\[
  \sum_{j \in Q} \Delta_{i,j}(\mathbf{b}_j) = V_i(\mathbf{b}) - S_i(\mathbf{b}) \ge 0.
\]
Furthermore, the unconstrained optimal allocation $\mathbf{x}^{\OPT}$ that maximizes liquid welfare assigns each query $j \in Q$ entirely to an advertiser with the highest value on query $j$, achieving optimal liquid welfare
\begin{equation}
\label{eq:lw_opt}
  \LW(\OPT) = \sum_{j \in Q} \max_{i \in A} v_{i,j}.
\end{equation}

\begin{definition}[Equilibrium and Undominated Bids]
\label{def:equilibrium}
Fix an instance $\mathcal{I}$ with advertisers $A = \{1, \dots, n\}$, queries $Q$, and values $\{v_{i,j}\}$.
\begin{enumerate}
  \item A bid profile $\mathbf{b} = (\mathbf{b}_1, \dots, \mathbf{b}_n) \in \R_{\ge 0}^{n \times |Q|}$ is an \emph{equilibrium} if each advertiser $i \in A$ satisfies its ROS constraint ($S_i(\mathbf{b}_i, \mathbf{b}_{-i}) \le V_i(\mathbf{b}_i, \mathbf{b}_{-i})$) and $\mathbf{b}_i \in \R_{\ge 0}^{|Q|}$ is an optimal solution to advertiser $i$'s problem~\eqref{eq:autobidder_opt} given $\mathbf{b}_{-i}$.
  \item Following \cite[Definition~5.2]{liaw2023efficiency}, a bid profile $\mathbf{b} \in \R_{\ge 0}^{n \times |Q|}$ is \emph{undominated} if each advertiser $i \in A$ satisfies its ROS constraint ($S_i(\mathbf{b}_i, \mathbf{b}_{-i}) \le V_i(\mathbf{b}_i, \mathbf{b}_{-i})$) and for every single query $j \in Q$, changing $b_{i,j}$ to any $b'_{i,j} \ge 0$ (while keeping all other bids fixed) either does not strictly increase advertiser $i$'s value on query $j$ or violates advertiser $i$'s ROS constraint. Every equilibrium bid profile is undominated.
\end{enumerate}
We denote the liquid welfare of an equilibrium (or undominated) bid profile $\mathbf{b}$ by
\[
  \LW(\ALG) = \sum_{i=1}^n V_i(\mathbf{b}) = \sum_{i=1}^n \sum_{j \in Q} x_{i,j}(\mathbf{b}_j) v_{i,j}.
\]
\end{definition}

\begin{definition}[(Liquid) Price of Anarchy]
\label{def:poa}
For a mechanism $\mathcal{M}$ and an instance $\mathcal{I} = (A, Q, \{v_{i,j}\})$, let $\mathrm{Eq}(\mathcal{I})$ denote the set of equilibria of $\mathcal{M}$ on $\mathcal{I}$. The \emph{(liquid) Price of Anarchy} ($\PoA$) of $\mathcal{M}$ is defined as
\[
  \PoA(\mathcal{M}) = \sup_{\mathcal{I}} \sup_{\mathbf{b} \in \mathrm{Eq}(\mathcal{I})} \frac{\LW(\OPT)}{\LW(\ALG)}.
\]
\end{definition}

\section{Proportional FPA Achieves \texorpdfstring{$\PoA \le 1.5$}{PoA <= 1.5}}
\label{sec:upper_bound}

In this section, we specialize to the two-bidder Proportional First-Price Auction ($\pFPA = \pFPA_1$, with $n = 2$ and $r = 1$) and prove that any equilibrium achieves at least $\frac{2}{3}$ of the optimal liquid welfare through four lemmas: (1) positivity of equilibrium bids on positive-value queries (Lemma~\ref{lem:pos_bid}); (2) a Lagrangian duality and complementary slackness characterization of equilibrium bidding (Lemma~\ref{lem:binding_ros}); (3) a first-order characterization relating values to equilibrium bids and dual multipliers $\mu_i$ on contested queries (Lemma~\ref{lem:foc}); and (4) a pointwise inequality bounding $\max(v_{1,j}, v_{2,j})$ on each query (Lemma~\ref{lem:pointwise}).

\begin{lemma}
\label{lem:pos_bid}
Let $\mathbf{b} = (\mathbf{b}_1, \mathbf{b}_2) \in \R_{\ge 0}^{2 \times |Q|}$ be any equilibrium of $\mathcal{M}_{\mathrm{prop}}$ on an instance $\mathcal{I} = (A, Q, \{v_{i,j}\})$ with $v_{i,j} \ge 0$. Then for every advertiser $i \in \{1, 2\}$ and query $j \in Q$, if $v_{i,j} > 0$, then $b_{i,j} > 0$.
\end{lemma}

\begin{proof}
Suppose for contradiction that $v_{i,j} > 0$ and $b_{i,j} = 0$ for some advertiser $i \in \{1, 2\}$ and query $j \in Q$, so $x_{i,j}(0, b_{-i,j}) = 0$ and $p_{i,j}(0, b_{-i,j}) = 0$. If advertiser $i$ deviates on query $j$ to $b'_{i,j} = \frac{1}{2}v_{i,j} > 0$, then $x'_{i,j} = x_{i,j}(b'_{i,j}, b_{-i,j}) > 0$ and $p'_{i,j} = p_{i,j}(b'_{i,j}, b_{-i,j}) \le b'_{i,j} x'_{i,j} = \frac{1}{2}v_{i,j} x'_{i,j}$. This deviation strictly increases advertiser $i$'s total expected value $V_i$ by $x'_{i,j} v_{i,j} > 0$ while increasing net surplus $V_i - S_i$ by $\Delta'_{i,j} = x'_{i,j} v_{i,j} - p'_{i,j} \ge \frac{1}{2}x'_{i,j} v_{i,j} > 0$ (preserving the ROS constraint $V_i - S_i \ge 0$), contradicting that $\mathbf{b}$ is an equilibrium.
\end{proof}

By introducing a Lagrange multiplier $\lambda_i \ge 0$ for advertiser $i$'s ROS constraint $V_i - S_i \ge 0$, advertiser $i$ maximizes the Lagrangian $V_i + \lambda_i(V_i - S_i)$. Dividing by $1 + \lambda_i > 0$ and setting $\mu_i = \frac{\lambda_i}{1 + \lambda_i} \in [0, 1)$, we obtain the decoupled per-query surrogate objective $f_{i,j}(b) = x_{i,j}(b, b_{-i,j})v_{i,j} - \mu_i p_{i,j}(b, b_{-i,j})$. The following lemma shows via KKT optimality that at any equilibrium, each advertiser $i$ has a dual multiplier $\mu_i \in [0, 1)$ satisfying complementary slackness $\mu_i(V_i(\mathbf{b}) - S_i(\mathbf{b})) = 0$ such that each equilibrium bid $b_{i,j}$ globally maximizes $f_{i,j}(b)$.

\begin{lemma}
\label{lem:binding_ros}
Let $\mathbf{b} = (\mathbf{b}_1, \mathbf{b}_2) \in \R_{\ge 0}^{2 \times |Q|}$ be any equilibrium of $\mathcal{M}_{\mathrm{prop}}$ on an instance $\mathcal{I} = (A, Q, \{v_{i,j}\})$ with $v_{i,j} \ge 0$. Then for each advertiser $i \in \{1, 2\}$, there exists a dual multiplier $\mu_i \in [0, 1)$ satisfying complementary slackness
\begin{equation}
\label{eq:comp_slack}
  \mu_i \bigl(V_i(\mathbf{b}) - S_i(\mathbf{b})\bigr) = \mu_i \sum_{j \in Q} \Delta_{i,j}(b_{1,j}, b_{2,j}) = 0
\end{equation}
such that for every query $j \in Q$, the equilibrium bid $b_{i,j} \ge 0$ globally maximizes the per-query surrogate objective
\begin{equation}
\label{eq:surrogate_obj}
  f_{i,j}(b) = x_{i,j}(b, b_{-i,j}) v_{i,j} - \mu_i p_{i,j}(b, b_{-i,j})
\end{equation}
over $b \ge 0$.
\end{lemma}

\begin{proof}
Fix an advertiser $i \in \{1, 2\}$, and let $Q_{-i}^- = \{j \in Q : b_{-i,j} = 0\}$ denote the set of queries that are uncontested by the competing advertiser $-i$. For any $j \in Q_{-i}^-$, we have $p_{i,j}(b, 0) = 0$ for all $b \ge 0$, and if $v_{i,j} > 0$ then $b_{i,j} > 0$ by Lemma~\ref{lem:pos_bid} so $x_{i,j}(b_{i,j}, 0) = 1$. Thus $b_{i,j}$ achieves value $v_{i,j}$ with zero payment and maximizes $f_{i,j}(b)$ over $b \ge 0$ for any $\mu_i \ge 0$, contributing total value $V_i^0 := \sum_{j \in Q_{-i}^-} v_{i,j} \ge 0$ and zero payment.

Now let $Q_{-i}^+ = \{j \in Q : b_{-i,j} > 0\}$ denote the set of queries on which advertiser $-i$ submits a positive bid. Advertiser $i$'s optimization problem~\eqref{eq:autobidder_opt} over the contested queries $Q_{-i}^+$ can be reparameterized by the winning probabilities $x_{i,j} = \frac{b_{i,j}}{b_{i,j} + b_{-i,j}} \in [0, 1)$ for $j \in Q_{-i}^+$, where $b_{i,j} = b_{-i,j}\frac{x_{i,j}}{1 - x_{i,j}}$ and $p_{i,j}(x_{i,j}) = b_{-i,j}\frac{x_{i,j}^2}{1 - x_{i,j}}$:
\[
  \max_{\mathbf{x}_i \in [0, 1)^{|Q_{-i}^+|}} \quad V_i(\mathbf{x}_i) = V_i^0 + \sum_{j \in Q_{-i}^+} v_{i,j} x_{i,j} \qquad \text{s.t.} \quad S_i(\mathbf{x}_i) - V_i(\mathbf{x}_i) = \sum_{j \in Q_{-i}^+} p_{i,j}(x_{i,j}) - V_i(\mathbf{x}_i) \le 0.
\]
Here the objective $V_i(\mathbf{x}_i)$ is linear, and since $\frac{\mathrm{d}^2 p_{i,j}}{\mathrm{d}x_{i,j}^2} = \frac{2b_{-i,j}}{(1 - x_{i,j})^3} > 0$ on $[0, 1)$, the constraint function $S_i(\mathbf{x}_i) - V_i(\mathbf{x}_i)$ is strictly convex on $[0, 1)^{|Q_{-i}^+|}$. If $v_{i,j} = 0$ for all $j \in Q_{-i}^+$, then $V_i(\mathbf{x}_i) \equiv V_i^0$ is constant and $\mu_i = 0$ directly satisfies Eqs.~\eqref{eq:comp_slack}--\eqref{eq:surrogate_obj}. Otherwise, there exists some $j_0 \in Q_{-i}^+$ with $v_{i,j_0} > 0$, so Slater's condition holds: setting $b'_{i,j_0} = \frac{1}{2}v_{i,j_0} > 0$ (which gives $\Delta'_{i,j_0} = \frac{1}{2}x'_{i,j_0}v_{i,j_0} > 0$) and $b'_{i,j} = \eta > 0$ on $j \in Q_{-i}^+ \setminus \{j_0\}$ for sufficiently small $\eta > 0$ yields an interior point $\mathbf{x}'_i \in (0, 1)^{|Q_{-i}^+|}$ with $S_i(\mathbf{x}'_i) - V_i(\mathbf{x}'_i) < 0$. By the Karush--Kuhn--Tucker (KKT) theorem, there exists a Lagrange multiplier $\lambda_i \ge 0$ satisfying complementary slackness $\lambda_i\bigl(V_i(\mathbf{x}_i) - S_i(\mathbf{x}_i)\bigr) = 0$ such that the equilibrium allocation $\mathbf{x}_i$ globally maximizes the Lagrangian
\[
  \mathcal{L}_i(\mathbf{x}_i; \lambda_i) = V_i(\mathbf{x}_i) + \lambda_i \bigl(V_i(\mathbf{x}_i) - S_i(\mathbf{x}_i)\bigr)
\]
over $[0, 1)^{|Q_{-i}^+|}$. Dividing $\mathcal{L}_i(\mathbf{x}_i; \lambda_i)$ by $1 + \lambda_i > 0$ and setting $\mu_i = \frac{\lambda_i}{1 + \lambda_i} \in [0, 1)$ gives Eq.~\eqref{eq:comp_slack} together with $\frac{\mathcal{L}_i}{1 + \lambda_i} = \sum_{j \in Q} f_{i,j}(b_{i,j})$, which separates across $j \in Q$ and implies that $b_{i,j}$ globally maximizes $f_{i,j}(b)$ over $b \ge 0$ for every $j \in Q$.
\end{proof}

\begin{lemma}
\label{lem:foc}
Let $\mathbf{b} = (\mathbf{b}_1, \mathbf{b}_2)$ be an equilibrium of $\mathcal{M}_{\mathrm{prop}}$ with dual multipliers $\mu_1, \mu_2 \in [0, 1)$ from Lemma~\ref{lem:binding_ros}. For any advertiser $i \in \{1, 2\}$ and any query $j \in Q$ with a strictly positive competing bid $a = b_{-i,j} > 0$, the equilibrium bid $b = b_{i,j} \ge 0$ satisfies
\begin{equation}
\label{eq:foc}
  v_{i,j} = \mu_i b \left( \frac{b}{a} + 2 \right).
\end{equation}
\end{lemma}

\begin{proof}
If $b = 0$, then $v_{i,j} = 0$ by Lemma~\ref{lem:pos_bid}, so both sides of Eq.~\eqref{eq:foc} equal $0$. Now suppose $b > 0$. By Lemma~\ref{lem:binding_ros}, $b$ maximizes $f_{i,j}(b) = v_{i,j} \frac{b}{a + b} - \mu_i \frac{b^2}{a + b}$ over $[0, \infty)$. Differentiating $f_{i,j}(b)$ with respect to $b$ gives
\begin{equation}
\label{eq:first_deriv}
  f'_{i,j}(b) = v_{i,j} \frac{a}{(a + b)^2} - \mu_i \frac{b(b + 2a)}{(a + b)^2} = \frac{v_{i,j} a - \mu_i(b^2 + 2ab)}{(a + b)^2}.
\end{equation}
Since $b > 0$ is an interior maximizer, we have $f'_{i,j}(b) = 0$, which yields $v_{i,j} a = \mu_i b(b + 2a)$, or equivalently $v_{i,j} = \mu_i b \left(\frac{b}{a} + 2\right)$.
\end{proof}

We now prove our main technical inequality: a pointwise upper bound on $\max(v_{1,j}, v_{2,j})$ for every query $j \in Q$.

\begin{lemma}
\label{lem:pointwise}
For any equilibrium bid pair $(b_1, b_2) \in \R_{\ge 0}^2$ on a single query with dual multipliers $\mu_1, \mu_2 \ge 0$, expected payments $p_1, p_2$, and expected surpluses $\Delta_1, \Delta_2$, the following pointwise inequality holds:
\begin{equation}
\label{eq:pointwise_bound}
  \max(v_1, v_2) \le \frac{3}{2}(p_1 + p_2) + (1 + 3\mu_1)\Delta_1 + (1 + 3\mu_2)\Delta_2.
\end{equation}
\end{lemma}

\begin{proof}
We consider two cases depending on whether the query is contested.

\paragraph{Case 1: Uncontested query ($b_1 = 0$ or $b_2 = 0$).}
If $b_1 = b_2 = 0$, then $v_1 = v_2 = 0$ by Lemma~\ref{lem:pos_bid}, so both sides of Eq.~\eqref{eq:pointwise_bound} equal $0$. If $b_1 > 0$ and $b_2 = 0$ (or symmetrically $b_2 > 0$ and $b_1 = 0$), then by Lemma~\ref{lem:pos_bid} and Eq.~\eqref{eq:payment} we have $v_2 = 0$, $x_1(b_1, 0) = 1$, $p_1 = p_2 = 0$, $\Delta_1 = v_1 \ge 0$, and $\Delta_2 = 0$; since $\mu_1 \ge 0$, $\max(v_1, v_2) = v_1 = \Delta_1 \le (1 + 3\mu_1)\Delta_1$, so Eq.~\eqref{eq:pointwise_bound} holds.

\paragraph{Case 2: Contested query ($b_1 > 0$ and $b_2 > 0$).}
Because the right-hand side of Eq.~\eqref{eq:pointwise_bound} is symmetric in the indices $1$ and $2$, it suffices to prove that
\begin{equation}
\label{eq:gap_D_def}
  D := \frac{3}{2}(p_1 + p_2) + (1 + 3\mu_1)\Delta_1 + (1 + 3\mu_2)\Delta_2 - v_1 \ge 0
\end{equation}
holds for all $b_1, b_2 > 0$ and $\mu_1, \mu_2 \ge 0$; index symmetry then gives the same bound for $v_2$. Let $z = b_1 / b_2 > 0$, so that $b_1 = z b_2$. Then the allocation probabilities and expected payments are
\[
  x_1 = \frac{b_1}{b_1 + b_2} = \frac{z}{z + 1}, \quad
  x_2 = \frac{b_2}{b_1 + b_2} = \frac{1}{z + 1}, \quad
  p_1 = x_1 b_1 = b_2 \frac{z^2}{z + 1}, \quad
  p_2 = x_2 b_2 = b_2 \frac{1}{z + 1},
\]
and by Lemma~\ref{lem:foc}, the values $v_1$ and $v_2$ are
\[
  v_1 = \mu_1 b_1 \left(\frac{b_1}{b_2} + 2\right) = b_2 \mu_1 z(z + 2), \qquad
  v_2 = \mu_2 b_2 \left(\frac{b_2}{b_1} + 2\right) = b_2 \mu_2 \left(\frac{1}{z} + 2\right) = b_2 \mu_2 \frac{2z + 1}{z}.
\]
Since $p_1, p_2, v_1, v_2$ (and hence $\Delta_i = x_i v_i - p_i$) all have a common factor of $b_2 > 0$, we may set $b_2 = 1$ without loss of generality.
Next, substituting $\Delta_1 = x_1 v_1 - p_1$ and $\Delta_2 = x_2 v_2 - p_2$ into Eq.~\eqref{eq:gap_D_def} and grouping the $-p_1$ and $-p_2$ terms with $\frac{3}{2}(p_1 + p_2)$ yields
\begin{align*}
  D &= \frac{3}{2}(p_1 + p_2) + (1 + 3\mu_1)(x_1 v_1 - p_1) + (1 + 3\mu_2)(x_2 v_2 - p_2) - v_1 \\
    &= \frac{1}{2}(p_1 + p_2) + \bigl[(1 + 3\mu_1)x_1 v_1 - v_1 - 3\mu_1 p_1\bigr] + \bigl[(1 + 3\mu_2)x_2 v_2 - 3\mu_2 p_2\bigr].
\end{align*}
We now express each of the three terms on the right-hand side as a rational function of $z$. First, summing the expected payments gives $\frac{1}{2}(p_1 + p_2) = \frac{z^2 + 1}{2(z + 1)}$. Second, substituting $x_1 = \frac{z}{z + 1}$, $p_1 = \frac{z^2}{z + 1}$, and $v_1 = \mu_1 z(z + 2)$ into the first bracket gives
\begin{align*}
  (1 + 3\mu_1)x_1 v_1 - v_1 - 3\mu_1 p_1
  &= \left(\frac{(1 + 3\mu_1)z}{z + 1} - 1\right) v_1 - \frac{3\mu_1 z^2}{z + 1} \\
  &= \frac{(3\mu_1 z - 1)\mu_1 z(z + 2) - 3\mu_1 z^2}{z + 1} \\
  &= \frac{3z^2(z + 2)\mu_1^2 - 2z(2z + 1)\mu_1}{z + 1}.
\end{align*}
Third, substituting $x_2 = \frac{1}{z + 1}$, $p_2 = \frac{1}{z + 1}$, and $v_2 = \mu_2 \frac{2z + 1}{z}$ into the second bracket gives
\begin{align*}
  (1 + 3\mu_2)x_2 v_2 - 3\mu_2 p_2
  &= \frac{(1 + 3\mu_2)v_2 - 3\mu_2}{z + 1} \\
  &= \frac{(1 + 3\mu_2)\mu_2(2z + 1) - 3\mu_2 z}{z(z + 1)} \\
  &= \frac{3(2z + 1)\mu_2^2 - (z - 1)\mu_2}{z(z + 1)}.
\end{align*}
Multiplying $D$ by the common denominator $z(z + 1) > 0$ and defining $F(z, \mu_1, \mu_2) := z(z + 1) D$, we obtain
\begin{equation}
\label{eq:F_decomp}
  F(z, \mu_1, \mu_2) = \frac{1}{2}z(z^2 + 1) + h_1(z, \mu_1) + h_2(z, \mu_2),
\end{equation}
where
\begin{align*}
  h_1(z, \mu_1) &= 3z^3(z + 2)\mu_1^2 - 2z^2(2z + 1)\mu_1, \\
  h_2(z, \mu_2) &= 3(2z + 1)\mu_2^2 - (z - 1)\mu_2.
\end{align*}
Completing the square in $\mu_1$ gives
\[
  h_1(z, \mu_1) = 3z^3(z + 2)\left(\mu_1 - \frac{2z + 1}{3z(z + 2)}\right)^2 - \frac{z(2z + 1)^2}{3(z + 2)} \ge -\frac{z(2z + 1)^2}{3(z + 2)}.
\]
Combining this lower bound with the first term of Eq.~\eqref{eq:F_decomp}, and noting that $3z + 4 = 2(z + 2) + z > 2(z + 2)$ for all $z > 0$, we obtain
\begin{align}
  \frac{1}{2}z(z^2 + 1) + h_1(z, \mu_1)
  &\ge \frac{z\bigl[3(z + 2)(z^2 + 1) - 2(2z + 1)^2\bigr]}{6(z + 2)} \nonumber \\
  &= \frac{z(3z^3 - 2z^2 - 5z + 4)}{6(z + 2)} \nonumber \\
  &= \frac{z(z - 1)^2(3z + 4)}{6(z + 2)} \ge \frac{1}{3}z(z - 1)^2. \label{eq:h1_combined}
\end{align}
To bound $h_2(z, \mu_2)$, observe that $z(z - 1) - (z - 1) = (z - 1)^2 \ge 0$, so $z - 1 \le z(z - 1)$ for all $z > 0$. Since $\mu_2 \ge 0$ and $3(2z + 1) = 6z + 3 > 6z$, completing the square in $\mu_2$ gives
\begin{align}
  h_2(z, \mu_2) = 3(2z + 1)\mu_2^2 - (z - 1)\mu_2
  &\ge 6z\mu_2^2 - z(z - 1)\mu_2 \nonumber \\
  &= 6z\left(\mu_2 - \frac{z - 1}{12}\right)^2 - \frac{1}{24}z(z - 1)^2 \ge -\frac{1}{24}z(z - 1)^2. \label{eq:h2_bound}
\end{align}
Summing Eqs.~\eqref{eq:h1_combined} and~\eqref{eq:h2_bound} yields
\[
  F(z, \mu_1, \mu_2) \ge \left(\frac{1}{3} - \frac{1}{24}\right)z(z - 1)^2 = \frac{7}{24}z(z - 1)^2 \ge 0
\]
for all $z > 0$ and $\mu_1, \mu_2 \ge 0$, which proves $D \ge 0$.
\end{proof}

\begin{theorem}
\label{thm:main}
In the two-bidder Proportional First-Price Auction ($\pFPA_1$), the Price of Anarchy is at most $1.5$.
\end{theorem}

\begin{proof}
Let $\mathbf{b} = (\mathbf{b}_1, \mathbf{b}_2)$ be any equilibrium of $\pFPA_1$. By Lemma~\ref{lem:binding_ros}, there exist dual multipliers $\mu_1, \mu_2 \in [0, 1)$ satisfying complementary slackness $\mu_i \bigl(V_i(\mathbf{b}) - S_i(\mathbf{b})\bigr) = 0$ for each $i \in \{1, 2\}$. Summing the pointwise inequality of Lemma~\ref{lem:pointwise} over all queries $j \in Q$, and using $\sum_{j \in Q} p_{i,j} = S_i(\mathbf{b})$, $\sum_{j \in Q} \Delta_{i,j} = V_i(\mathbf{b}) - S_i(\mathbf{b})$, and the ROS constraint $S_i(\mathbf{b}) \le V_i(\mathbf{b})$, we obtain
\begin{align*}
  \LW(\OPT) = \sum_{j \in Q} \max(v_{1,j}, v_{2,j})
  &\le \sum_{i=1}^2 \left[ \frac{3}{2} \sum_{j \in Q} p_{i,j} + (1 + 3\mu_i) \sum_{j \in Q} \Delta_{i,j} \right] \\
  &= \sum_{i=1}^2 \left[ \frac{3}{2} S_i(\mathbf{b}) + \bigl(V_i(\mathbf{b}) - S_i(\mathbf{b})\bigr) + 3\mu_i \bigl(V_i(\mathbf{b}) - S_i(\mathbf{b})\bigr) \right] \\
  &= \sum_{i=1}^2 \left[ V_i(\mathbf{b}) + \frac{1}{2} S_i(\mathbf{b}) \right] \le \frac{3}{2} \sum_{i=1}^2 V_i(\mathbf{b}) = \frac{3}{2}\,\LW(\ALG). \qedhere
\end{align*}
\end{proof}

\section{Two-Bidder Lower Bound: \texorpdfstring{$\PoA \ge 1.5$}{PoA >= 1.5}}
\label{sec:lower_bound}

We now show that our upper bound of $1.5$ for $\mathcal{M}_{\mathrm{prop}}$ is not merely tight for Proportional FPA, but cannot be improved by any two-bidder mechanism satisfying mild assumptions. As discussed in the introduction, our construction adapts the lower bound of Liaw, Mehta, and Perlroth~\cite[Theorem~6.2]{liaw2023efficiency}.

Following the general mechanism framework of Liaw, Mehta, and Perlroth~\cite[Section~6]{liaw2023efficiency}, let $\mathcal{M} = (\pi, c)$ be any two-bidder single-slot auction mechanism defined by an allocation rule $\pi = (\pi_1, \pi_2) : \R_{\ge 0}^2 \to [0, 1]^2$ with $\pi_1(b_1, b_2) + \pi_2(b_1, b_2) \le 1$ and a conditional price rule $c = (c_1, c_2) : \R_{\ge 0}^2 \to \R_{\ge 0}^2$ (so bidder $i$ wins with probability $\pi_i(b_1, b_2)$ and pays expected cost $p_i(b_1, b_2) = \pi_i(b_1, b_2) c_i(b_1, b_2)$). We assume $\mathcal{M}$ satisfies three properties:
\begin{enumerate}
  \item \textbf{Monotone Allocation:} For each $i \in \{1, 2\}$, $\pi_i(b_1, b_2)$ is non-decreasing in $b_i$, a zero bid does not win against a positive bid ($\pi_i(0, b_{-i}) = 0$ for $b_{-i} > 0$), and the maximum uncontested allocation $\pi_0 := \max_{b > 0} \pi_1(b, 0) > 0$ is attained.
  \item \textbf{Anonymity:} For all $(b_1, b_2) \in \R_{\ge 0}^2$, $\pi_1(b_1, b_2) = \pi_2(b_2, b_1)$ and $c_1(b_1, b_2) = c_2(b_2, b_1)$.
  \item \textbf{Monotone Conditional Price at a Tie:} An uncontested positive bid pays $c_i(b_i, 0) = 0$, and there exists a symmetric positive bid $B > 0$ with $v := c_1(B, B) = c_2(B, B) > 0$ and $\pi^* := \pi_1(B, B) = \pi_2(B, B) > 0$ such that bidding strictly above $B$ against $B$ does not decrease the conditional price ($c_i(b'_i, B) \ge c_i(B, B) = v$ for all $b'_i > B$ with $\pi_i(b'_i, B) > \pi^*$).
\end{enumerate}
Note that these assumptions (which are implied by Assumption~6.1 of \cite{liaw2023efficiency}) are satisfied by \emph{every} standard anonymous auction mechanism---including $\pFPA$, $\FPA$, $\SPA$, $\rFPA(\alpha)$ for any $\alpha \ge 1$, $\rTruth(\alpha)$ for any $\alpha > 1$ \cite{mehta2022auction, liaw2023efficiency}, and proportional style mechanisms $\pi_i = \frac{b_i^r}{b_1^r + b_2^r}$ with $r > 0$---all of which satisfy $\pi_0 = 1$. The uncontested zero payment assumption $c_i(b_i, 0) = 0$ in Property~(3) is necessary for the existence of exact equilibria when some valuations are zero: without it, a sole positive bidder on a query ($b_{i,j} > 0$, $b_{-i,j} = 0$) would want to lower $b_{i,j} \to 0^+$ to reduce spend while retaining the full allocation, precluding equilibrium existence.

\begin{theorem}
\label{thm:lower_bound}
Let $\mathcal{M} = (\pi, c)$ be any two-bidder mechanism satisfying Properties (1)--(3) above, let $B > 0$, $v = c_1(B, B) > 0$, $\pi^* = \pi_1(B, B) \in (0, 1/2]$, and let $B_0 > 0$ satisfy $\pi_1(B_0, 0) = \pi_0$. Then for every $\epsilon \in (0, v/2)$, on the three-query instance $\mathcal{I}_\epsilon$ with values
\[
\begin{aligned}
  v_{1,1} &= \epsilon, \quad v_{2,1} = v - \epsilon \quad (\text{Query~1}), \\
  v_{1,2} &= \pi^*(v - \epsilon)/\pi_0, \quad v_{2,2} = 0 \quad (\text{Query~2}), \\
  v_{1,3} &= 0, \quad v_{2,3} = \pi^*\epsilon/\pi_0 \quad (\text{Query~3}),
\end{aligned}
\]
the bid profile $\mathbf{b} = ((B, B), (B_0, 0), (0, B_0))$ is an equilibrium of $\mathcal{M}$ with
\begin{equation}
\label{eq:universal_lb_ratio}
  \frac{\LW(\OPT)}{\LW(\ALG)} = \frac{1}{2\pi^*} + \frac{1}{2\pi_0} - \frac{\epsilon}{2\pi^* v} \ge \frac{3}{2} - \frac{\epsilon}{2\pi^* v},
\end{equation}
and hence $\PoA(\mathcal{M}) \ge 1.5 - \frac{\epsilon}{2\pi^* v}$.
\end{theorem}

Since Theorem~\ref{thm:lower_bound} holds for every $\epsilon \in (0, v/2)$, taking $\epsilon \to 0^+$ shows that $\PoA(\mathcal{M}) \ge 1.5$, so the $1.5$ upper bound of Theorem~\ref{thm:main} cannot be improved.

While the intuitive 2-query construction in the introduction illustrates the welfare gap, Theorem~\ref{thm:lower_bound} requires a third query to establish equilibrium existence across general mechanisms $\mathcal{M}$ (such as SPA). Under mechanisms where bidding above a tie does not increase the conditional price ($c_i(b_i', B) = v$), setting $v_{2,1} = v$ would allow bidder 2 to deviate to $b_2' > B$ and win the query entirely without violating their ROS constraint. Setting $v_{2,1} = v - \epsilon$ ensures that winning more of query 1 incurs a strict per-win deficit of $\epsilon$, preventing upward deviations. Query 3 then provides an uncontested surplus of $\pi^*\epsilon$, funding this baseline deficit at the tie $(B, B)$ without providing any excess budget to bid higher.

\begin{proof}[Proof of Theorem~\ref{thm:lower_bound}]
By anonymity and feasibility, $2\pi^* = \pi_1(B, B) + \pi_2(B, B) \le 1$, so $\pi^* \in (0, 1/2]$. By Property~(1), a zero bid does not win against a positive bid ($\pi_i(0, b_{-i}) = 0$ for $b_{-i} > 0$), and $\pi_1(B_0, 0) = \pi_0 = \max_{b > 0} \pi_1(b, 0)$, so $\pi_1(b, 0) \le \pi_0$ for all $b > 0$. Under the bid profile $\mathbf{b} = ((B, B), (B_0, 0), (0, B_0))$:
\begin{itemize}
  \item \textbf{On Query~1 ($b_{1,1} = b_{2,1} = B$):} Each bidder $i \in \{1, 2\}$ wins with probability $\pi_i(B, B) = \pi^*$ and pays conditional price $c_i(B, B) = v$, for an expected payment of $p_i(B, B) = \pi^* v$. Bidder~1 obtains expected value $\pi^* \epsilon$ and net surplus $\Delta_{1,1} = \pi^*(\epsilon - v) < 0$, while Bidder~2 obtains expected value $\pi^*(v - \epsilon)$ and net surplus $\Delta_{2,1} = -\pi^*\epsilon < 0$.
  \item \textbf{On Query~2 ($b_{1,2} = B_0, b_{2,2} = 0$):} Bidder~1 wins with probability $\pi_1(B_0, 0) = \pi_0$ and pays $c_1(B_0, 0) = 0$, obtaining value $\pi_0 v_{1,2} = \pi^*(v - \epsilon)$ and net surplus $\Delta_{1,2} = \pi^*(v - \epsilon)$. Bidder~2 has $\pi_2(B_0, 0) = 0$ by Property~(1), so $\Delta_{2,2} = 0$.
  \item \textbf{On Query~3 ($b_{1,3} = 0, b_{2,3} = B_0$):} By anonymity, $\pi_2(0, B_0) = \pi_1(B_0, 0) = \pi_0$ and $c_2(0, B_0) = c_1(B_0, 0) = 0$. Bidder~2 obtains value $\pi_0 v_{2,3} = \pi^*\epsilon$ and net surplus $\Delta_{2,3} = \pi^*\epsilon$. Bidder~1 has $\pi_1(0, B_0) = 0$ by Property~(1), so $\Delta_{1,3} = 0$.
\end{itemize}
Summing over all three queries, Bidder~1 has total value $V_1(\mathbf{b}) = \pi^*\epsilon + \pi^*(v - \epsilon) = \pi^* v$ and total spend $S_1(\mathbf{b}) = \pi^* v$, and Bidder~2 has total value $V_2(\mathbf{b}) = \pi^*(v - \epsilon) + \pi^*\epsilon = \pi^* v$ and total spend $S_2(\mathbf{b}) = \pi^* v$, so both ROS constraints hold with equality ($V_1 = S_1$ and $V_2 = S_2$).

We now verify that neither bidder can profitably deviate:
\begin{enumerate}
  \item \textbf{Bidder~1 cannot profitably deviate:} For any deviation $(b'_{1,1}, b'_{1,2}, b'_{1,3}) \in \R_{\ge 0}^3$, since $\pi_1(b, 0) \le \pi_0$ for all $b > 0$ and $c_1(b, 0) = 0$ for all $b > 0$, Queries~2 and~3 together provide value at most $\pi_0 v_{1,2} + v_{1,3} = \pi^*(v - \epsilon)$ and zero payment on uncontested queries, contributing at most $\pi^*(v - \epsilon)$ to Bidder~1's net surplus $V'_1 - S'_1$. If $b'_{1,1} > B$ with $\pi_1(b'_{1,1}, B) > \pi^*$, then by Property~(3) we have $c_1(b'_{1,1}, B) \ge c_1(B, B) = v > \epsilon$, so Bidder~1's net deficit on Query~1 strictly exceeds $\pi^*(v - \epsilon)$:
  \[
  \begin{aligned}
    p_1(b'_{1,1}, B) - \epsilon \pi_1(b'_{1,1}, B)
    &= \pi_1(b'_{1,1}, B)\bigl(c_1(b'_{1,1}, B) - \epsilon\bigr) \\
    &\ge \pi_1(b'_{1,1}, B)(v - \epsilon) > \pi^*(v - \epsilon),
  \end{aligned}
  \]
  which implies $S'_1 - V'_1 > \pi^*(v - \epsilon) - \pi^*(v - \epsilon) = 0$, violating Bidder~1's ROS constraint. Otherwise (if $\pi_1(b'_{1,1}, B) \le \pi^*$, which by Property~(1) includes all $b'_{1,1} \le B$), Bidder~1's total value is at most $V'_1 \le \pi^*\epsilon + \pi^*(v - \epsilon) = V_1(\mathbf{b})$. Thus $\mathbf{b}_1 = (B, B_0, 0)$ is a best response for Bidder~1.
  \item \textbf{Bidder~2 cannot profitably deviate:} By anonymity, $\pi_2(0, b) = \pi_1(b, 0) \le \pi_0$, so Queries~2 and~3 together contribute at most $v_{2,2} + \pi_0 v_{2,3} = \pi^*\epsilon$ to Bidder~2's net surplus $V'_2 - S'_2$. If $b'_{2,1} > B$ with $\pi_2(B, b'_{2,1}) > \pi^*$, then by Properties~(2)--(3) we have $c_2(B, b'_{2,1}) = c_1(b'_{2,1}, B) \ge v > v - \epsilon$, so Bidder~2's net deficit on Query~1 strictly exceeds $\pi^*\epsilon$:
  \[
  \begin{aligned}
    p_2(B, b'_{2,1}) - (v - \epsilon) \pi_2(B, b'_{2,1})
    &= \pi_2(B, b'_{2,1})\bigl(c_2(B, b'_{2,1}) - (v - \epsilon)\bigr) \\
    &\ge \pi_2(B, b'_{2,1})\epsilon > \pi^*\epsilon,
  \end{aligned}
  \]
  which implies $S'_2 - V'_2 > \pi^*\epsilon - \pi^*\epsilon = 0$, violating Bidder~2's ROS constraint. Otherwise ($\pi_2(B, b'_{2,1}) \le \pi^*$), Bidder~2's total value satisfies $V'_2 \le \pi^*(v - \epsilon) + \pi^*\epsilon = V_2(\mathbf{b})$. Thus $\mathbf{b}_2 = (B, 0, B_0)$ is a best response for Bidder~2.
\end{enumerate}
Finally, the equilibrium liquid welfare is $\LW(\ALG) = V_1(\mathbf{b}) + V_2(\mathbf{b}) = 2\pi^* v$. Since $\epsilon \in (0, v/2)$, we have $v - \epsilon > \epsilon$, so assigning Query~1 to Bidder~2, Query~2 to Bidder~1, and Query~3 to Bidder~2 achieves $\LW(\OPT) = (v - \epsilon) + \frac{\pi^*(v - \epsilon)}{\pi_0} + \frac{\pi^*\epsilon}{\pi_0} = (v - \epsilon) + \frac{\pi^* v}{\pi_0}$. Since $\pi^* \in (0, 1/2]$ and $\pi_0 \in (0, 1]$, we have
\[
  \frac{\LW(\OPT)}{\LW(\ALG)} = \frac{(v - \epsilon) + \pi^* v / \pi_0}{2\pi^* v} = \frac{1}{2\pi^*} + \frac{1}{2\pi_0} - \frac{\epsilon}{2\pi^* v} \ge \frac{3}{2} - \frac{\epsilon}{2\pi^* v},
\]
completing the proof.
\end{proof}

\section{Multi-Bidder Analysis: \texorpdfstring{$2 - \Theta(1/n)$}{2 - Theta(1/n)} Price of Anarchy}
\label{sec:n_bidders}

We now turn to the general multi-bidder setting with $n \ge 2$ advertisers and resolve the open question from Liaw, Mehta, and Perlroth~\cite{liaw2023efficiency}: while Theorem~6.2 of \cite{liaw2023efficiency} proved that under mild assumptions, any $n$-bidder mechanism has $\PoA \ge \frac{2n+4}{n+4} = 2 - \frac{4}{n+4} = 2 - O(1/n)$, we prove that the $2n$-Proportional First-Price Auction ($\pFPA_{2n}$, i.e., $\mathcal{M}_{\mathrm{prop}}^{(r)}$ with exponent $r = 2n$) achieves $\PoA(\pFPA_{2n}) \le 2 - \frac{1}{4n + 1} = 2 - \Omega(1/n)$ across all undominated bid profiles (and thus all equilibria) for every $n \ge 2$. Our analysis adapts the three-step undominated-bid and value-plus-spend convex-combination approach of Liaw, Mehta, and Perlroth~\cite[Section~5]{liaw2023efficiency} (which builds on Mehta~\cite{mehta2022auction}), extending it from two bidders to $n \ge 2$ bidders via Jensen's inequality.

Following \cite[Lemma~5.5]{liaw2023efficiency}, we first establish a lower bound on any undominated bid in $\pFPA_r$ for any exponent $r > 0$ by comparing marginal value and marginal spend.

\begin{lemma}
\label{lem:n_undom}
Fix $n \ge 2$ and $r > 0$. For any undominated bid profile $\mathbf{b} \in \R_{\ge 0}^{n \times |Q|}$ of $\pFPA_r$, every advertiser $i \in [n]$ and query $j \in Q$ with winning probability $x_{i,j} \in [0, 1]$ satisfies
\begin{equation}
\label{eq:n_undom_bound}
  b_{i,j} \ge \frac{r(1 - x_{i,j})}{r(1 - x_{i,j}) + 1} v_{i,j}.
\end{equation}
\end{lemma}

\begin{proof}
If $v_{i,j} = 0$ or $x_{i,j} = 1$, the right-hand side of Eq.~\eqref{eq:n_undom_bound} equals $0$, so Eq.~\eqref{eq:n_undom_bound} holds trivially since $b_{i,j} \ge 0$. Now suppose $v_{i,j} > 0$ and $x_{i,j} \in [0, 1)$. Note first that we cannot have $b_{i,j} = 0$ (which gives $x_{i,j} = 0$), since deviating to $b'_{i,j} = \frac{1}{2}v_{i,j} > 0$ would yield $x'_{i,j} > 0$, strictly increasing advertiser $i$'s value on query $j$ by $x'_{i,j} v_{i,j} > 0$ while increasing net surplus $V_i - S_i$ by $\Delta'_{i,j} \ge \frac{1}{2}x'_{i,j} v_{i,j} > 0$. Thus $b_{i,j} > 0$, and since $x_{i,j} < 1$, we also have $B_{-i,j} := \sum_{k \ne i} b_{k,j}^r > 0$. Differentiating $x_{i,j}(b) = \frac{b^r}{b^r + B_{-i,j}}$ with respect to $b = b_{i,j} > 0$ yields
\[
  \frac{\partial x_{i,j}}{\partial b_{i,j}}
  = \frac{r b_{i,j}^{r-1} B_{-i,j}}{(b_{i,j}^r + B_{-i,j})^2}
  = \frac{r}{b_{i,j}} x_{i,j}(1 - x_{i,j}) > 0.
\]
Consequently, the marginal changes in advertiser $i$'s expected value $V_i$ and expected spend $S_i$ with respect to $b_{i,j}$ are
\[
  \frac{\partial V_i}{\partial b_{i,j}} = v_{i,j} \frac{\partial x_{i,j}}{\partial b_{i,j}}, \qquad
  \frac{\partial S_i}{\partial b_{i,j}} = x_{i,j} + b_{i,j} \frac{\partial x_{i,j}}{\partial b_{i,j}}.
\]
If $\frac{\partial V_i}{\partial b_{i,j}} > \frac{\partial S_i}{\partial b_{i,j}}$, then increasing $b_{i,j}$ by a sufficiently small $\delta > 0$ strictly increases advertiser $i$'s value on query $j$ while also increasing $V_i - S_i$ (preserving the ROS constraint $V_i \ge S_i$), contradicting that $\mathbf{b}$ is undominated. Hence an undominated bid must satisfy $\frac{\partial V_i}{\partial b_{i,j}} \le \frac{\partial S_i}{\partial b_{i,j}}$, which gives
\[
  v_{i,j} \le b_{i,j} \left( 1 + \frac{x_{i,j}}{b_{i,j} (\partial x_{i,j} / \partial b_{i,j})} \right) = b_{i,j} \left( 1 + \frac{1}{r(1 - x_{i,j})} \right) = b_{i,j} \frac{r(1 - x_{i,j}) + 1}{r(1 - x_{i,j})},
\]
proving Eq.~\eqref{eq:n_undom_bound}.
\end{proof}

Next, we use Lemma~\ref{lem:n_undom} together with Jensen's inequality across the $n - 1$ non-optimal bidders to lower-bound the aggregate expected spend $S^{(j)} := \sum_{k=1}^n p_{k,j}$ on query $j$.

\begin{lemma}
\label{lem:n_spend}
Fix a query $j \in Q$, let $i^* \in \arg\max_{i \in [n]} v_{i,j}$ be an optimal bidder on query $j$ with value $v^*_j = v_{i^*,j}$, bid $b^*_j = b_{i^*,j}$, and winning probability $x^*_j = x_{i^*,j} \in (0, 1]$. For any undominated bid profile of $\pFPA_r$, the aggregate spend $S^{(j)} = \sum_{k=1}^n p_{k,j}$ on query $j$ satisfies
\begin{equation}
\label{eq:n_spend_bound}
  S^{(j)} \ge v^*_j \cdot g(x^*_j), \quad \text{where} \quad
  g(x) := \frac{r(1 - x)}{r(1 - x) + 1} \left( x + \frac{(1 - x)^{1 + 1/r}}{(n - 1)^{1/r} x^{1/r}} \right)
\end{equation}
for all $x \in (0, 1]$ (with $g(1) = \lim_{x \to 1^-} g(x) = 0$ by continuity).
\end{lemma}

\begin{proof}
If $x^*_j = 1$, then $g(1) = 0$ and $S^{(j)} \ge 0$ holds trivially. Suppose $x^*_j \in (0, 1)$, so that $\sum_{k \ne i^*} b_{k,j} > 0$ and $S^{(j)} = \sum_{k=1}^n x_{k,j} b_{k,j}$. By Eq.~\eqref{eq:alloc}, for every $k \in [n]$ we have $\frac{x_{k,j}}{x^*_j} = \bigl(\frac{b_{k,j}}{b^*_j}\bigr)^r$, so $b_{k,j} = b^*_j \bigl(\frac{x_{k,j}}{x^*_j}\bigr)^{1/r}$. Substituting this relation into $S^{(j)}$ yields
\[
  S^{(j)} = \sum_{k=1}^n x_{k,j}^{1 + 1/r} \frac{b^*_j}{(x^*_j)^{1/r}} = b^*_j \left( x^*_j + \frac{\sum_{k \ne i^*} x_{k,j}^{1 + 1/r}}{(x^*_j)^{1/r}} \right).
\]
Since $z \mapsto z^{1 + 1/r}$ is convex on $[0, \infty)$ and $\sum_{k \ne i^*} x_{k,j} = 1 - x^*_j$, applying Jensen's inequality to the $n - 1$ terms $\{x_{k,j}\}_{k \ne i^*}$ gives
\[
  \sum_{k \ne i^*} x_{k,j}^{1 + 1/r} \ge (n - 1) \left( \frac{\sum_{k \ne i^*} x_{k,j}}{n - 1} \right)^{1 + 1/r} = \frac{(1 - x^*_j)^{1 + 1/r}}{(n - 1)^{1/r}}.
\]
Combining this inequality with Lemma~\ref{lem:n_undom} ($b^*_j \ge \frac{r(1 - x^*_j)}{r(1 - x^*_j) + 1} v^*_j$) yields $S^{(j)} \ge v^*_j g(x^*_j)$.
\end{proof}

Following \cite[Lemma~5.6]{liaw2023efficiency} and \cite{mehta2022auction}, we now take a convex combination of the optimal bidder's expected value and the aggregate expected spend to bound the Price of Anarchy in terms of $g(x)$.

\begin{lemma}
\label{lem:n_duality}
Let $c := \inf_{x \in (0, 1)} \frac{g(x)}{1 - x} > 0$. For any undominated bid profile of $\pFPA_r$, the liquid welfare satisfies $\LW(\OPT) \le \bigl(1 + \frac{1}{c}\bigr) \LW(\ALG)$, and thus
\begin{equation}
\label{eq:poa_c_bound}
  \PoA(\pFPA_r) \le 1 + \frac{1}{c}.
\end{equation}
\end{lemma}

\begin{proof}
Since each advertiser satisfies $V_i(\mathbf{b}) \ge S_i(\mathbf{b})$, the liquid welfare $\LW(\ALG) = \sum_{i=1}^n V_i(\mathbf{b})$ satisfies both $\LW(\ALG) \ge \sum_{j \in Q} x^*_j v^*_j$ and $\LW(\ALG) \ge \sum_{i=1}^n S_i(\mathbf{b}) = \sum_{j \in Q} S^{(j)} \ge \sum_{j \in Q} v^*_j g(x^*_j)$ by Lemma~\ref{lem:n_spend}. By definition of $c = \inf_{x \in (0, 1)} \frac{g(x)}{1 - x}$, for every $x \in (0, 1)$ we have $\frac{g(x)}{1 - x} \ge c$ and multiplying by $1 - x > 0$ gives $g(x) \ge c(1 - x)$; at $x = 1$, both sides equal $0$ ($g(1) = 0 = c(1 - 1)$), so $g(x^*_j) \ge c(1 - x^*_j)$ holds for all $x^*_j \in (0, 1]$. Setting $\gamma = \frac{c}{c + 1} \in (0, 1)$ and $\eta = \frac{1}{c + 1} = 1 - \gamma$, we have $\gamma = \eta c$. Taking the convex combination $\LW(\ALG) = \gamma\,\LW(\ALG) + \eta\,\LW(\ALG)$ yields
\begin{align*}
  \LW(\ALG)
  &\ge \sum_{j \in Q} v^*_j \bigl(\gamma x^*_j + \eta g(x^*_j)\bigr) \ge \sum_{j \in Q} v^*_j \bigl(\eta c x^*_j + \eta c(1 - x^*_j)\bigr) \\
  &= \eta c \sum_{j \in Q} v^*_j = \gamma\,\LW(\OPT),
\end{align*}
so $\frac{\LW(\OPT)}{\LW(\ALG)} \le \frac{1}{\gamma} = \frac{\gamma + \eta}{\gamma} = 1 + \frac{1}{c}$.
\end{proof}

\begin{lemma}
\label{lem:n_root}
For any $n \ge 2$ and $r > 0$, the minimum $c = \inf_{x \in (0, 1)} \frac{g(x)}{1 - x}$ is the unique positive root of the characteristic equation
\begin{equation}
\label{eq:char_poly}
  (n - 1) c^{r + 1} + c - r = 0.
\end{equation}
In particular, when $r = 2n$, we have $c > 1 + \frac{1}{4n}$ for all $n \ge 2$.
\end{lemma}

\begin{proof}
For $x \in (0, 1)$, let $y = \frac{1 - x}{x} \in (0, \infty)$ (so $x = \frac{1}{y + 1}$ and $1 - x = \frac{y}{y + 1}$) and $A := (n - 1)^{-1/r} > 0$. Dividing Eq.~\eqref{eq:n_spend_bound} by $1 - x > 0$ and multiplying numerator and denominator by $y + 1$, we can write $h(y) := \frac{g(x)}{1 - x}$ as
\[
  h(y) = \frac{r}{r(1 - x) + 1}\bigl(x + A(1 - x)y^{1/r}\bigr) = \frac{r\bigl(1 + A y^{1 + 1/r}\bigr)}{(r + 1)y + 1}.
\]
Differentiating $h(y)$ with respect to $y > 0$ using the quotient rule gives
\begin{align*}
  h'(y)
  &= \frac{(r + 1) A y^{1/r}\bigl[(r + 1)y + 1\bigr] - r(r + 1)\bigl(1 + A y^{1 + 1/r}\bigr)}{\bigl((r + 1)y + 1\bigr)^2} \\
  &= (r + 1) \frac{A y^{1/r}(y + 1) - r}{\bigl((r + 1)y + 1\bigr)^2}.
\end{align*}
Since $y \mapsto A y^{1/r}(y + 1)$ is strictly increasing from $0$ to $\infty$ on $(0, \infty)$, $h(y)$ attains its unique global minimum on $(0, \infty)$ at the unique point $y^* > 0$ satisfying $A (y^*)^{1/r}(y^* + 1) = r$. Defining $z := A(y^*)^{1/r} > 0$, this optimality condition gives $z(y^* + 1) = r$, i.e., $y^* + 1 = \frac{r}{z}$ and $A(y^*)^{1 + 1/r} = z y^*$. Substituting these identities into $h(y^*)$ yields
\[
  c = h(y^*) = \frac{r(1 + z y^*)}{(r + 1)y^* + 1} = \frac{z(y^* + 1)(1 + z y^*)}{z(y^* + 1)y^* + y^* + 1} = z.
\]
Since $c = z = (n - 1)^{-1/r}(y^*)^{1/r}$, we have $y^* = (n - 1)c^r$; substituting this into $c(y^* + 1) = r$ gives $(n - 1)c^{r + 1} + c - r = 0$.

Now set $r = 2n$ and consider $P(c) := (n - 1)c^{2n + 1} + c - 2n$, which is strictly increasing on $c > 0$. Evaluating $P$ at $c_0 := 1 + \frac{1}{4n}$, since $\bigl(1 + \frac{1}{4n}\bigr)^{4n} < e$ (so $\bigl(1 + \frac{1}{4n}\bigr)^{2n} < \sqrt{e}$) and $1 + \frac{1}{4n} \le \frac{9}{8} = 1.125$ for all $n \ge 2$, we have
\[
\begin{aligned}
  \left(1 + \frac{1}{4n}\right)^{2n + 1}
  &= \left(1 + \frac{1}{4n}\right)\left(1 + \frac{1}{4n}\right)^{2n} \\
  &< 1.125 \sqrt{e} < 1.855.
\end{aligned}
\]
Consequently, for all $n \ge 2$,
\[
\begin{aligned}
  P(c_0) &< 1.855(n - 1) + 1.125 - 2n \\
  &= -0.145 n - 0.730 < 0 = P(c),
\end{aligned}
\]
which implies $c > c_0 = 1 + \frac{1}{4n}$.
\end{proof}

Combining Lemma~\ref{lem:n_duality} and Lemma~\ref{lem:n_root} immediately yields our main multi-bidder theorem.

\begin{theorem}
\label{thm:n_bidders}
For any $n \ge 2$ bidders and any undominated bid profile (and hence any equilibrium), the $2n$-Proportional First-Price Auction ($\pFPA_{2n}$) achieves a Price of Anarchy of at most
\begin{equation}
\label{eq:n_bidders_poa}
  \PoA(\pFPA_{2n}) \le 1 + \frac{1}{c} < 1 + \frac{1}{1 + \frac{1}{4n}} = 2 - \frac{1}{4n + 1}.
\end{equation}
\end{theorem}

\bibliographystyle{plainnat}
\bibliography{refs}

\end{document}